\documentclass[prb,twocolumn,superscriptaddress,showpacs,floatfix,longbibliography]{revtex4-2}
\usepackage{appendix}
\usepackage{tabularx}
\usepackage{amsmath}
\usepackage[margin=1in,letterpaper]{geometry}
\usepackage[final]{hyperref}
\hypersetup{
	colorlinks=true,
	linkcolor=blue,
	citecolor=blue,
	filecolor=magenta,
	urlcolor=blue
}
\usepackage{graphicx}
\usepackage{amssymb}
\usepackage{mathtools}
\usepackage{float,bm}
\usepackage{braket}
\newtheorem{theorem}{Theorem}
\newtheorem{lemma}{Lemma}

\newtheorem{proof}{Proof}

\begin{document}

\title{Higher-order Liouvillian exceptional points in the dissipative dynamics of quadratic fermions}

\author{Mingtao Xu}
\affiliation{Laboratory of Quantum Information, University of Science and Technology of China, Hefei 230026, China}
\author{Wei Yi}
\email{wyiz@ustc.edu.cn}
\affiliation{Laboratory of Quantum Information, University of Science and Technology of China, Hefei 230026, China}
\affiliation{Anhui Province Key Laboratory of Quantum Network, University of Science and Technology of China, Hefei 230026, China}
\affiliation{CAS Center For Excellence in Quantum Information and Quantum Physics, Hefei 230026, China}
\affiliation{Hefei National Laboratory, University of Science and Technology of China, Hefei 230088, China}
\date{\today}

\begin{abstract}
We propose a general class of open fermionic models where quadratic Liouvillians governing the dissipative dynamics feature { analytically characterized higher-order exceptional points (EPs).}
Invoking the formalism of third quantization, we show that, among the multiple EPs of Liouvillian, an EP with its order approaching the system size arises as { the dominant modes of the system at long times,} 
leading to a gapless Liouvillian spectrum.
By introducing perturbations, in the form of many-body quantum-jump processes, these higher-order EPs break down, leading to finite Liouvillian gaps with fractional power-law scalings.
While the power-law scaling is a signature of the higher-order EP, its explicit form is sensitively dependent on the many-body perturbation. Finally, we discuss { the long-time dynamics} which can serve as detectable signals for the higher-order Liouvillian EPs. 
\end{abstract}

\maketitle

\section{Introduction}

Exceptional points (EPs) are spectral degeneracies in non-Hermitian systems at which both eigenvalues and their corresponding eigenstates coalesce~\cite{ep1966,ep1998,ep1999,MiriExceptional,ep4,BergholtzExceptional,DingNon-Hermitian}. The associated non-Hermitian Hamiltonian, in the form of non-Hermitian matrix, becomes defective and can be cast into the Jordan form, wherein the dimensions of the Jordan blocks indicate the orders of the EPs~\cite{AshidaNonHermitian}. 
While EPs, and indeed novel features of non-Hermitian physics in general, have mostly been discussed in the context of linear or single-particle systems, they also arise in nonlinear or many-body quantum settings. 
For instance, nonlinear exceptional structures and their dynamic consequences have recently been reported in coupled resonators~\cite{nonlinearH1,nonlinearH2,nonlinearH3} and dissipative Rydberg vapors~\cite{XieChiral, ZhangExceptional,KopciuchLiouvillian}, where the nonlinearity either derives from frequency-dependent nonlinear gain or many-body effects in the thermal Rydberg gas. 

On the other hand, non-Hermiticity also arises naturally in quantum open systems. More explicitly, the Markovian dynamics of an open quantum system is typically governed by the Lindblad master equation~\cite{Breuer}
\begin{equation}
         \frac{\mathrm{d}\rho}{\mathrm{d}t} =\mathcal{L}\rho= -i[H,\rho]+\sum_\mu \Big[L_\mu\rho L_\mu^\dagger - \frac{1}{2}\{L_\mu^\dagger L_\mu,\rho\}\Big],
         \label{Lindblad1}
\end{equation}
where  $\mathcal{L}$ is referred to as the Liouvillian, $\rho$ is the density matrix, $H$ is the coherent Hamiltonian, and $\{L_\mu\}$ is a set of quantum jump operators depicting the system's coupling with its environment. 
By imposing post selection and requiring the absence of quantum jump processes described by $\sum_\mu L_\mu\rho L_\mu^\dagger$, the Liouvillian dynamics is reduced to one driven by an effective non-Hermitian Hamiltonian $H_{\text{eff}}=H-\frac{i}{2}\sum_{\mu}L_\mu^\dagger L_\mu$. 
Alternatively, one can adopt a different view by vectorizing the density matrix, such that the Liouvillian itself is mapped to a non-Hermitian matrix~\cite{MingantiSpectral}. Such a view underlies recent studies of Liouvillian EPs~\cite{SunEncircling,MingantiQuantum,MingantiHybrid,KumarNear,KhandelwalChiral,PavlovTopological,SunChiral}, where exceptional structures emerge in the Liouvillian eigenspectrum with experimentally observable outcome~\cite{GaoPhotonic,ChenDecoherence,ChenQuantum,WuExperimental}.
Nevertheless, in these early studies, the order of the Liouvillian EP is often limited to 2 or 3, whereas higher-order EPs have been well-studied in generic non-Hermitian systems, and are shown to have richer implications~\cite{DelplaceSymmetry,YoshidaHopf,SayyadRealizing,GohsrichExceptional,ShiralievaMulti}. It is therefore desirable to engineer higher-order Liouvillian EPs in a quantum many-body setting~\cite{ArkhipovLiouvillian2,ArkhipovGenerating,NakanishiPT}. 

In this work, we study a general class of hybrid-Liouvillians, where quantum jump processes are partially post selected. Focusing on quadratic open systems of fermions, where the Liouvillian eigenspectrum can be analytically solved within the third-quantization framework~\cite{ProsenThird,ProsenSpectral}, we demonstrate the presence of higher-order Liouvillian EPs in the system. 
{ Under partial post selection, the hybrid-Liouvillian is trace decreasing, so the unnormalized density matrix vanishes in the long-time limit. We therefore focus on the normalized density matrix, whose long-time behavior is governed by the slowest-decaying Liouvillian modes.}
{ 
We refer to these modes as the quasi-steady states, in analogy with the eigenstates of a non-Hermitian Hamiltonian with the largest imaginary eigenvalue component, which also dominate the long-time dynamics and regarded as the steady states of the non-Hermitian dynamics~\cite{AzconaMagic}.}
Remarkably, the quasi-steady state itself corresponds to the highest-order Liouvillian EP of the system, whose order scales with the system size. 
The resulting Liouvillian spectrum is gapless, leading to an algebraic decay toward the quasi-steady state in the dissipative dynamics. 
We show that the nature of these higher-order EPs can be revealed by introducing many-body quantum jump processes as perturbations. 
For instance, upon introducing the perturbation, the degeneracy is lifted at the highest-order EP corresponding to the quasi-steady state, but with the perturbed Liouvillian eigenvalues and the resulting Liouvillian gap exhibiting perturbation-dependent fractional power-law scalings. 
This gives rise to exponentially fast asymptotic dynamics toward the quasi-steady state, whose perturbation-dependent relaxation time should facilitate the design of dynamic detection schemes for the higher-order Liouvillian EPs. 
By engineering higher-order Liouvillian EPs in many-body quantum open systems, our work paves the way for 
further exploration of non-Hermitian physics in the quantum many-body regime.

\section{Dissipative dynamics under a hybrid-Liouvillian}

\subsection{Hybrid master equation under the third quantization}

We consider a general quantum open system described by the Lindblad master equation (\ref{Lindblad1}).
{ Given a set of quantum jump operators \(\{L_\mu\}\), we introduce a complementary set \(\{L_\mu'\}\), such that
\begin{equation}
\sum_{\mu}\left(L_\mu^\dagger L_\mu+L_\mu'^\dagger L_\mu'\right)=cI,
\end{equation}
for a positive constant \(c>0\). For the finite-size fermionic systems considered in this work, \(\sum_\mu L_\mu^\dagger L_\mu\) is a bounded positive operator. Hence, the constant \(c\) can be chosen sufficiently large so that \(cI-\sum_\mu L_\mu^\dagger L_\mu\) is positive semidefinite. This positive semidefinite operator can always be factorized as \(\sum_\mu L_\mu'^\dagger L_\mu'\), leading to the complementary jump operators.}

The full Liouvillian superoperator then becomes
\begin{equation}
    \mathcal{L}\rho=-i[H,\rho]+\sum_\mu\left(L_\mu\rho L_\mu^\dagger +L_\mu'\rho L_\mu'^\dagger \right) - c\rho.
\end{equation}
Upon post selecting only those trajectories that undergo no quantum jumps of the type $\{L_\mu'\}$, the resulting subensemble dynamics is described by the hybrid master equation
\begin{equation}
    \frac{\mathrm{d}\rho}{\mathrm{d}t}=\mathcal{L}_H\rho=-i[H,\rho]+\sum_\mu L_\mu\rho L_\mu^\dagger-c\rho.
    \label{Lindblad2}
\end{equation}
Here the last term $-c\rho$ shifts the entire Liouvillian spectrum by $-c$, but does not alter the eigenstates of the Liouvillian.

{ We emphasize that Eq.~(\ref{Lindblad2}) is not trace preserving. The density matrix $\rho(t)$ should be understood as an unnormalized density matrix associated with the post-selected ensemble in which complementary quantum jumps $L_\mu'$ do not occur. Its trace,
\begin{equation}
    P_{\rm post}(t)=\mathrm{Tr}\rho(t),
\end{equation}
is the probability of obtaining a trajectory in the post-selected ensemble.
Physical observables within such an ensemble are hence
\begin{equation}
    \langle \tilde O\rangle(t)
    =
    \frac{\mathrm{Tr}[O\rho(t)]}{\mathrm{Tr}\rho(t)}
    =
    \mathrm{Tr}[O\tilde\rho(t)],
\end{equation}
where we define the normalized density matrix $\tilde \rho(t) = \rho(t)/ P_{\rm post}(t)$. 
In the long-time limit, among the sectors with largest real part of the Liouvillian eigenvalue, the sector with the largest Jordan block dominates the long-time behavior.
}

Under these premises, we study a quadratic open system of fermions on $n$ sites, which can be conveniently analyzed under the formalism of the third quantization~\cite{ProsenThird}. In particular, a general quadratic Hamiltonian $H$ and a set of linear quantum jump operators $L_{\mu}$ can be expressed in the Majorana basis
\begin{align}
    H &=\sum_{j,k=1}^{2n}w_jh_{jk}^Mw_k,\label{eq:w1}\\
    L_\mu&=\sum_{j=1}^{2n}l_{\mu,j}^M w_j,\label{eq:w2}
\end{align}
where the matrix $h^M$ is antisymmetric $h^M=-(h^M)^T$, and $w_i$ are the operators for Majorana fermions satisfying the anti-commutation relation $\{w_j,w_k\}=\delta_{j,k}$.
For the convenience of discussion, we choose $c=\mathrm{Tr}M^M$, where the coefficient matrix $M_{jk}^M=\sum_\mu l^M_{\mu,j}(l^M_{\mu,k})^*$. 

The hybrid master equation Eq.~(\ref{Lindblad2}) can be decomposed into three parts, with $\mathcal{L}_{H}=\mathcal{L}_0+\mathcal{L}_c+\sum_\mu \mathcal{L}_{\mu}$. Here 
the unitary $\mathcal{L}_0\rho=-i[H,\rho]$, the quantum jumps $\sum_\mu \mathcal{L}_\mu\rho=\sum_\mu L_\mu\rho L_\mu^\dagger$, and the overall shift $\mathcal{L}_c\rho=-\mathrm{Tr}M^M\rho$.
Under the third quantization,
the unitary part is mapped to $\hat{\mathcal{L}}_0=-2i\sum_{j,k}^{2n} \hat f_j^\dagger h_{jk}^M \hat f_k$, where $\hat f_i$ ($\hat f_i^\dagger$) is the annihilation (creation) operator for adjoint fermions satisfying the anti-commutation relations $\{ \hat f_j, \hat f_k\} = 0$ and $\{ \hat f_j,  \hat f_k^\dagger\}=\delta_{j,k}$.
Likewise, we have  $\hat{\mathcal{L}}_\mu= \sum_{j,k=1}^{2n}M^M_{jk}\hat{\mathcal{L}}_{j,k}$, where
\begin{align}
    &\hat{\mathcal{L}}_{j,k}=e^{i\pi \hat N}\times\nonumber\\
    &\bigg(\hat{f}_j^\dagger \hat{f}_k^\dagger-\frac{1}{2}\hat{f}_j^\dagger{f}_k+\frac{1}{2}\hat{f}_j\hat{f}_k^\dagger
    -\frac{1}{2}\hat{f}_k^\dagger\hat{f}_j + \frac{1}{2}\hat{f}_k \hat{f}_j^\dagger - \hat{f}_j \hat{f}_k\bigg),
\end{align}
with the adjoint fermion number operator given by $\hat{N}=\sum_j  \hat f_j^\dagger  \hat f_j$. And the shift term becomes $\hat{\mathcal{L}}_c = -\mathrm{Tr}M^M$.
As a result, $\mathcal{L}_{H}$ is mapped to the operator $\hat{\mathcal{L}}_{H}=\hat{\mathcal{L}}_{0}+\sum_\mu\hat{\mathcal{L}}_{\mu}+\hat{\mathcal{L}}_{c}$ in the adjoint-fermion space.
Since the adjoint-fermion { parity is conserved (with $[\hat{\mathcal{L}}_{H},e^{i\pi\hat{N}}]=0$)},
the Fock space $\mathcal{K}$ decomposes into a direct sum $\mathcal{K}=\mathcal{K}^+\oplus \mathcal{K}^-$, where $\mathcal{K}^\pm$ denote the even- and odd-parity subspaces, respectively. { Since the two parity sectors evolve independently and the physical observables considered here are parity even, the odd-parity sector does not enter the expectation values analyzed below. We thus restrict our attention to the even-parity sector $\mathcal{K}^+$}. Hence, we set $e^{i\pi\hat N}=1$, and obtain the Liouvillian in the third-quantized representation 
\begin{equation}
    \hat {\mathcal{L}}_{H} = \begin{pmatrix}
        \hat{\bm{f}} ^\dagger & \hat{\bm f}
    \end{pmatrix}\begin{pmatrix}
        -Z & Y\\
        Y^T & Z^T
    \end{pmatrix}\begin{pmatrix}
        \hat{\bm f}\\
        \hat {\bm{f}}^\dagger
    \end{pmatrix}-\mathrm{Tr}M^M,
\end{equation}
where $Z = ih^M+(M^M)^r$, $Y = i(M^M)^i$, and $\hat{\bm f}=(\hat f_1,\hat f_2,\cdots,\hat f_{2n})$. 

The superscripts $r$ and $i$ denote the real and imaginary parts of the matrices, respectively.
Following Ref.~\cite{ProsenThird}, we introduce $4n$ adjoint Majorana fermions $\hat a_{2j-1}=\frac{1}{\sqrt{2}}(\hat f_j + \hat f_j^\dagger)$ and $\hat a_{2j}=\frac{i}{\sqrt{2}}(\hat f_j-\hat f_j^\dagger)$. In terms of these operators, the Liouvillian $\hat{\mathcal{L}}_{H}$ is cast into a quadratic form
\begin{equation}
    \hat{\mathcal{L}}_{H}=\sum_{jk}\hat a_j A_{jk}\hat a_k- \mathrm{Tr}M^M,
\end{equation}
where the $4n\times 4n$ shape matrix $A$ is
\begin{align}
A &=-ih^M\otimes I-(M^M)^i\otimes\sigma_x - (M^M)^r\otimes \sigma_y.
\label{shapematrix1}
\end{align}
In the case that $A$ is diagonalizable, there exist $4n$ linearly independent eigenvectors with corresponding eigenvalues $\pm \beta_j$ (dubbed rapidities with $j= 1,\cdots,2n$).
It follows that the Liouvillian can be brought into a diagonal form
\begin{equation}
      \hat{\mathcal{L}}_{H} = -2\sum_{j=1}^{2n}\beta_j \hat b_j' \hat{b}_j +\left(\sum_{j=1}^{2n}\beta_j-\mathrm{Tr}M^M\right),
\end{equation}
where the fermion operators $\hat{b}'_j$ and $\hat{b}_j$ satisfy
$\{\hat b_j,\hat b_k\}=0$, $\{\hat b_j',\hat b_k'\}=0$, and $\{\hat b_j,\hat b_k'\}=\delta_{jk}$.
Given $2n$ binary integers $\nu_j\in\{0,1\}$, the Liouvillian spectrum can be expressed as
\begin{equation}
\label{spectrum}
    \lambda_\nu = -2\sum_{j=1}^{2n}\nu_j\beta _j+\sum_{j=1}^{2n}\beta _j-\mathrm{Tr}M^M.
\end{equation}

While the diagonalizable case presented above has been extensively studied~\cite{SongNonHermitian,OkumaQuantum,LiEngineering,YangLiouvillian}, the nondiagonalizable scenario is rarely explored but gives rise to exactly solvable Liouvillian EPs, as we show below.

\subsection{Nondiagonalizable shape matrix: higher-order Liouvillian EPs} 

As a special case of nondiagonalizable shape matrix, we consider a general quadratic Hamiltonian $H=\sum_{j\neq k}^nc_j^\dagger h_{jk}c_k$, where $h$ is the Hamiltonian matrix in the fermion basis, and
the fermion operators here are related to the Majorana fermions in Eqs.~(\ref{eq:w1}) and (\ref{eq:w2}), with
$c_{j}=\frac{1}{\sqrt{2}}( w_{2j-1} -i w_{2j})$ and $c_{j}^\dagger=\frac{1}{\sqrt{2}}(w_{2j-1}+i w_{2j})$. 
We also consider linear quantum jump operators
$L_\mu=\sum_{j=1}^n l_{\mu,j}c_j$, so that the 
corresponding matrices in the Majorana basis are
\begin{align}
    h^M &=\frac{1}{2}( {h}^r\otimes\sigma_y +ih^i\otimes I_2), \\
    M^M &=\frac{1}{2}M\otimes(I_2-\sigma_y),
\end{align}
where $M_{ i j } = \sum _ { \mu } l _ { \mu , i } l _ { \mu , j } ^ { * }$.
The shape matrix $A$ can then be written as
\begin{equation}
    A=T_1\otimes I_2+T_2\otimes i\sigma_y,
    \label{shapematrix}
\end{equation}
where the $2n\times 2n$ matrices $T_1$ and $T_2$ are
\begin{align}
     T_1 &= \begin{pmatrix}
           -\frac{1}{2}h^{i} & \frac{1}{2}iM\\
            -\frac{1}{2}iM^T & -\frac{1}{2}h^{i}
       \end{pmatrix}, \\
       T_2&= \begin{pmatrix}
           -\frac{1}{2}h^{r}  & -\frac{1}{2}M \\
     -\frac{1}{2}M^T& -\frac{1}{2} h^{r}
       \end{pmatrix}.
\end{align}
Under a unitary $U$ that transforms the Pauli matrices according to $\sigma_{x,y,z}\to \sigma_{y,z,x}$, the shape matrix $A$ is brought into a block-diagonal form
\begin{equation}
    UAU^\dagger = T_1\otimes I_2+T_2\otimes i\sigma_z=\begin{pmatrix}
        T_+ & \\
        & T_-
    \end{pmatrix},
\end{equation}
where $T_\pm=T_1\pm iT_2$ are respectively the lower and upper triangular block matrices
\begin{align}
    T_+&= T_1+iT_2 = \begin{pmatrix}
        -\frac{i}{2}h^T & \\
        -iM^T & -\frac{i}{2}h^T
    \end{pmatrix},\label{T_p}\\
    T_-&=T_1-iT_2 = \begin{pmatrix}
\frac{i}{2}h&  iM\\
       & \frac{i}{2}h
    \end{pmatrix}.
    \label{T_m}
\end{align}
Since $T_+=-(T_-)^T$, the characteristic polynomial of the shape matrix factorizes into the characteristic polynomials of $\frac{i}{2}h$, namely,
\begin{equation}
    p_A(\beta) = [p_{\frac{i}{2}h}(\beta)]^2[p_{\frac{i}{2}h}(-\beta)]^2,
\end{equation}
where $p_A(\beta)=\det[A-\beta I_{4n}]$, and $p_{\frac{i}{2}h}(\beta)=\det[\frac{i}{2}h-\beta I_{n}]$. 


Importantly, while the eigenvalues of the shape matrix $A$, which correspond to the rapidities, can be easily constructed from those of $\frac{i}{2}h$, the triangular block matrix $T_\pm$ is nondiagonalizable (see Appendix \ref{appendixA}). 
Specifically, denoting the eigenvalues of $\frac{i}{2}h$ as $\beta'_j$, each $\beta'_j$ corresponds to a Jordan block of order $n_j$, with $n_j\in\{1,2\}$.
Denoting $K$ as the total number of Jordan blocks of $T_+$, we have $\sum_{j=1}^K n_j=2n$.
It follows that the quadratic form for the Liouvillian is given by
\begin{align}
        \hat{\mathcal{L}}_{H,+} =& -2\sum_{j=1}^{n}\sum_{k=1}^{3-n_j}\bigg(\sum_{l=1}^{n_j}\beta'_{j} \hat b_{j,k,l}' \hat{b}_{j,k,l}\notag\\&+\sum_{l=1}^{n_j-1}\hat{b}_{j,k,l+1}'\hat b_{j,k,l}\bigg)
        +2\sum_{j=1}^{n}\beta'_j-\mathrm{Tr}M.
    \label{Liouvillian}
\end{align}
Here $\hat b_{j,k,l}$ and $\hat b_{j,k,l}'$ are operators for the adjoint fermions.
Since each $\beta'_j$ is associated with a nontrivial Jordan block of order $n_j=2$ or two trivial Jordan blocks of order $n_j=1$, the subscripts $k$ and $l$ are determined by $n_j$: for $n_j = 1$, we have $k = 1, 2$ and $l = 1$; whereas for $n_j = 2$, we have $k = 1$ and $l = 1, 2$.
The total number of operators $\hat b_{j,k,l}$ is thus $\sum_{j=1}^n\sum_{k=1}^{3-n_j}\sum_{l=1}^{n_j}1=2n$. The $4n$ operators $\hat b_{j,k,l}$ and $\hat b_{j,k,l}'$ introduced here satisfy
\begin{align}
 &   \{\hat b_{j,k,l},\hat b_{j',k',l'}\}=0,\quad \{\hat b_{j,k,l}',\hat b_{j',k',l'}'\}=0,\notag \\
 &\{\hat b_{j,k,l}, \hat b_{j',k',l'}'\}=\delta_{j,j'}\delta_{k,k'}\delta_{l,l'}.
\end{align}

It is noteworthy that the second term on the right-hand side of Eq.~(\ref{Liouvillian}) is off-diagonal, which vanishes for all trivial Jordan blocks $n_j=1$,  has no contribution to the Liouvillian spectrum.
The Liouvillian eigenspectrum is then given by
\begin{equation}
      \lambda_\nu =- 2\sum_{j=1}^{n}\sum_{k=1}^{3-n_j}\nu_{j,k}\beta' _j+2\sum_{j=1}^{n}\beta_j'-\mathrm{Tr}M,
      \label{spectrum2}
\end{equation}
where $\nu_{j,k}$ take integer values in $\{0,\cdots,n_j\}$, and the total number of $\nu_{j,k}$ equals $K$ (the total number of Jordan blocks of $T_+$).

From Eq.~(\ref{spectrum2}), we conclude that the Liouvillian spectrum $\{\lambda_\nu\}$ is always gapless, since $\mathrm{Re}\lambda_\nu=-\mathrm{Tr}M$. 
{ Thus the long-time behavior of the normalized dynamics is controlled by the largest Jordan block of the hybrid-Liouvillian.}

It is generally difficult to write down the Jordan form of the Liouvillian spectrum explicitly. However, based on the spectral theorem~\cite{ProsenSpectral}, one can determine the size of the largest Jordan block of the Liouvillian associated with a particular configuration $\nu$ in $\lambda_\nu$, which reads
\begin{equation}
    1+\sum_{j=1}^{n}\sum_{k=1}^{3-n_j}(n_j-\nu_{j,k})\nu_{j,k}.
    \label{size}
\end{equation}
Evaluation of the expression above gives the orders of the Liouvillian EPs.

Since $n_j\in\{1,2\}$ and $\nu_{j,k}\in\{0,\cdots,n_j\}$, the summand in Eq.~(\ref{size})
is nonvanishing and equal to unity, if and only if $n_j=2$ and $\nu_{j,k}=1$.
Let $K_1$ and $K_2$ denote the numbers of Jordan blocks of $T_+$ with $n_j=1$ and $n_j=2$, respectively.
According to Eq.~(\ref{size}), the size of the largest Jordan block is equal to one plus the number of nonvanishing terms in the summand. 
Hence, EPs of the Liouvillian can be of any order ranging from 2 to $1+K_2$.

In the next section, we present a concrete example where $K_1=0$ and $K_2=K=n$, so that the maximum order of the Liouvillian EP is $n+1$, which scales with the system size $n$. 
We note that the high-order EPs originate from a combinatorial stacking of identical defective single-particle $2\times2$ Jordan blocks, which is underpinned by the fermion anti-commutation relations and hence a collective many-body effect.

{ We also emphasize that the parity-odd sector also contains higher-order Liouvillian EPs. In the odd-parity sector, we may set $e^{i\pi\hat N}=-1$ and there is a global sign change in the contribution of jump-induced quadratic terms in Eq.~(\ref{shapematrix1}). Therefore, the off-diagonal Jordan structure in the shape matrix is the same for the odd-parity sector.}

\section{A solvable example}

\subsection{Higher-order Liouvillian EP}

\begin{figure}
    \centering
    \includegraphics[width=1\linewidth]{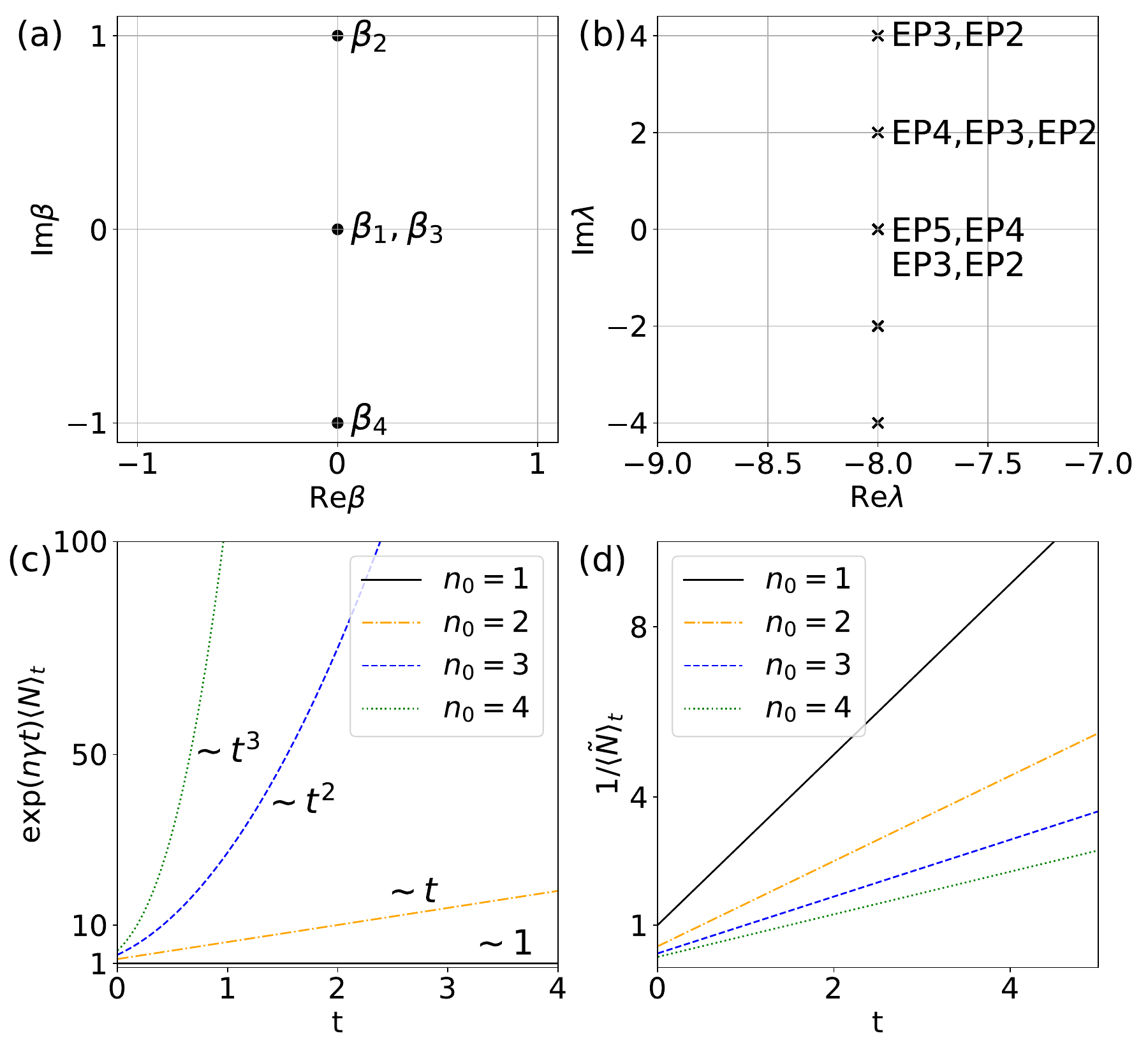}
    \caption{Liouvillian eigenspectrum and particle number dynamics for the dissipative fermion chain in Eqs.~(\ref{example_1}) and (\ref{example_2}). We set the system parameters as $t=1$, $\gamma=2$, and $n=4$. (a)(b) The rapidities and hybrid-Liouvillian spectrum, respectively, the hybrid-Liouvillian spectrum { has a common real part $-n\gamma$} and is highly degenerate. (c)(d) Dynamics of $\exp(n\gamma t)\braket{N}_t$ and $1/\braket{\tilde{N}}_t$ for different initial states $\rho(0,n_0)$ characterized by $n_0$, respectively.}
    \label{Fig1}
\end{figure}

In this section, we consider a concrete model with $n$ sites, whose Hamiltonian and jump operators are given by
\begin{align}
    H&=-t\sum_{j}c_j^\dagger c_{j+1}+\text{H.c.},\label{example_1}\\
    L_j&=\sqrt{\gamma}c_j,\quad j=1,2,\cdots, n.
    \label{example_2}
\end{align}
The hybrid-Liouvillian then takes the form
\begin{equation}
    \mathcal{L}_H\rho=-i[H,\rho]+\gamma\sum_i c_j\rho c_j^\dagger - n\gamma\rho,
    \label{eq:Lex}
\end{equation}
and the coefficient matrix $M=\gamma I_n$. We thus obtain
\begin{equation}
    T_+=-\frac{1}{2}i\begin{pmatrix}
        h & 0 \\
        2\gamma I_n & h
    \end{pmatrix},\quad T_-=-T_+^T.
\end{equation}
While our results here are independent of the boundary conditions, for simplicity, we adopt the periodic boundary condition below to exploit the translational symmetry of the system.
This allows us to write down the Fourier transform of matrix $T_+$
\begin{equation}
    T_+(q) = -it\cos q-\begin{pmatrix}
       0 & 0\\
       i\gamma & 0
    \end{pmatrix}.
\end{equation}
Note that for any given quasimomentum $q\in (0,2\pi]$, $T_+(q)$ is not diagonalizable.
Therefore, each $\beta'_j=-it\cos q_j$ (with $q_j=\frac{2\pi j}{n}$) is associated with a Jordan block of size $n_j=2$. The third-quantized representation of the Liouvillian is
\begin{align}
            \hat{\mathcal{L}}_{H} =& 2it\sum_{j=1}^{n}\cos\frac{2\pi j}{n}(\hat b_{j,1,1}' \hat{b}_{j,1,1}+ \hat b_{j,1,2}' \hat{b}_{j,1,2})\notag\\
    &  -2 \sum_j \hat{b}_{j,1,2}'\hat b_{j,1,1} - n\gamma.
\end{align}
The Liouvillian eigenspectrum is thus given by Eq.~(\ref{spectrum2}) with $n_j=2$
\begin{align}
 \lambda_\nu = 2it\sum_{j=1}^{n}\nu_{j,1}\cos \frac{2\pi j}{n}-n\gamma,
 \end{align}
where $\nu_{j,1}\in\{0,1,2\}$. We therefore have $K_1=0$, $K_2=n$. According to Eq.~(\ref{size}),
the maximum order of a Liouvillian EP is $n+1$, which is reached when $\nu_{j,1}=1$ for all $j$. The EP is associated with the Liouvillian eigenvalue $\lambda_0=-n\gamma $. We therefore conclude that, in this example, the Liouvillian spectrum is highly degenerate and hosts a large number of EPs whose orders range from $2$ to $n+1$. The highest-order Liouvillian EP corresponds to the quasi-steady state of the system.


In Fig.~\ref{Fig1}(a)(b), we show the rapidities and Liouvillian spectrum for $t=1$, $\gamma=2$, and $n=4$, respectively, wherein the order of the EPs are labelled. The eigenvalues in the $4^{4}$-dimensional Liouvillian space collapse to five discrete points in the complex plane: $-8$, $-8\pm2i$ and $-8\pm 4i$.
At $\lambda=-8$ in particular, there exists a 5th order Liouvillian EP (labelled EP5 in the figure), along with multiple EPs of orders $2$, $3$ and $4$, respectively. The EP5 is also the quasi-steady state of the system.

For further analysis, we introduce the generalized left and right eigenmatrices of Liouvillian $l_{\nu,k}$ and $r_{\nu,k}$, which satisfy the generalized characteristic equations (the Jordan chains)~\cite{MingantiQuantum}
{ 
\begin{align}
&(\mathcal{L}_H-\lambda_\nu I) r_{\nu,0}=  (   \mathcal{L}_H^\dagger-\lambda_\nu^* I) l_{\nu,m_\nu-1}=0, \label{Jordanchain1}\\
&(\mathcal{L}_H-\lambda_\nu I) r_{\nu,k}=\chi_{\nu,k} r_{\nu,k-1}, \label{Jordanchain2}\\
&(\mathcal{L}_H^\dagger-\lambda_\nu^* I) l_{\nu,k-1}=\chi_{\nu,k}^* l_{\nu,k},
\label{Jordanchain3}
\end{align}}
and the biorthogonal relations $\mathrm{Tr}(l_{\nu,k}r_{\nu',k'})=\delta_{\nu,\nu'}\delta_{k,k'}$. 
In the above expressions, we denote $m_\nu$ as the size of the Jordan block corresponding to the eigenvalue $\lambda_\nu$ under configuration $\nu$. { Here the coefficients \(\chi_{\nu,k}\) depend on the normalization convention of the generalized eigenmatrices.}
In particular, for the largest Jordan block with eigenvalue $\lambda_0=-n\gamma$, the left and right eigenmatrices are
\begin{align}
      l_{0,k}& = \sum_{\bm\eta _k}\ket{\bm\eta_k}\bra{\bm\eta_k}, \\
      r_{0,k} &= \frac{1}{C_n^k}\sum_{\bm\eta_k}\ket{\bm\eta_k}\bra{\bm\eta_k},
      \label{eigenmatrices}
\end{align}
where $\ket{\bm\eta_k}=c_{\eta_1}^\dagger c_{\eta_2}^\dagger\cdots c_{\eta_k}^\dagger\ket{0}$ is a $k$-particle state ($k=0,1\cdots,n$) and we denote $\bm \eta_k=(\eta_1,\eta_2,\cdots,\eta_k)$ with $\eta_1<\eta_2<\cdots< \eta_k$. The right eigenmatrices $r_{0,k}$ thus describes a maximally mixed state of all $k$-particle states. 
{ Based on the generalized eigenmatrices defined in Eq.~(\ref{eigenmatrices}), for the largest Jordan block, we have $\chi_{0,k}=\gamma k$ and the Jordan chains
\begin{equation}
    (\mathcal{L}_H+n\gamma) r_{0,k}=\gamma k r_{0,k-1}, \quad k=1,2,\cdots,n.
\end{equation}
}

Important features of the Liouvillian spectrum are manifest in the system dynamics. Here we focus on the dynamics of the particle number operator $N=\sum_j c_j^\dagger c_j$.
The formal solution of Eq.~(\ref{Lindblad2}) is given by~\cite{MingantiSpectral}
{ 
\begin{align}
      \rho(t) &= \sum_\nu e^{\lambda_\nu t}\sum_{k=0}^{m_\nu-1}\left[\sum_{j=0}^{m_\nu-1-k}\frac{t^j}{j!}c_{\nu,k+j}\prod_{p=k+1}^{k+j}\chi_{\nu,p}\right]r_{\nu,k},
      \label{formalsolution}
\end{align}}
where the coefficient $c_{\nu,k}=\mathrm{Tr}[\rho(0)l_{\nu,k}]$ depends on the initial state, { and we denote $\prod_{k+1}^k\chi_{\nu,p}=1$. For the largest Jordan block, we have $\prod_{p=k+1}^{k+j}\chi_{0,p}=\gamma^j\frac{(k+j)!}{k!}$.
It follows that 
\begin{align}
P_{\text{post}}&=\mathrm{Tr}\rho(t)=e^{-n\gamma t}\sum_{k=0}^n c_{0,k}(1+\gamma t)^k \label{eq:trace},\\
    \braket{N}_t&=\mathrm{Tr}N\rho(t)\notag  \\
 & = e^{-n\gamma t}\sum_{k=1}^n c_{0,k}k(1+\gamma t)^{k-1},\label{eq:N}
\end{align}}
where we have used $\mathrm{Tr}Nr_{0,k}=k$, and $\mathrm{Tr}N r_{\nu,k} = \mathrm{Tr}r_{\nu,k} =0$, for $\nu\neq 0$ (see Appendix~\ref{appendixB}). 
Note that the dynamics is solely determined by the eigenvalue $\lambda_0=-n\gamma$ and the corresponding quasi-steady state eigenmatrices. 

According to Eq.~(\ref{eq:N}), apart from a global exponential decay, the dynamics of the total particle number exhibit a polynomial scaling in time, with the maximum order given by $n-1$, related to the order of the Liouvillian EP. 
To highlight the gapless nature of the spectrum,  we normalize the density matrix $\tilde\rho(t)=\rho(t)/\mathrm{Tr}\rho(t)$ to eliminate the global exponential decay.
{ The dynamics of the total particle number under $\tilde{\rho}(t)$ is 
\begin{align}
    \braket{\tilde{N}}_t=&\mathrm{Tr}N\tilde{\rho}(t)\notag \\
=&\frac{\sum_{k=1}^n c_{0,k}k(1+\gamma t)^{k-1}}{\sum_{k=0}^n c_{0,k}(1+\gamma t)^{k}},
\end{align}
which, for a given initial state and in the long-time limit, scales as
\begin{align}
    \braket{\tilde{N}}_t\sim &\frac{c_{0,n_0}n_0(1+\gamma t)^{n_0-1}}{ c_{0,n_0}(1+\gamma t)^{n_0}}
   =\frac{n_0}{1+\gamma t}\sim\frac{n_0}{\gamma t}.
\end{align}}
Here $n_0$ is the maximum value of $k$ in the summations, determined by the maximum particle number of the initial state.

The analysis above is numerically confirmed in Fig.~\ref{Fig1}(c)(d). Here, the initial state is set to $\rho(0,n_0)=\prod_{i=1}^{n_0}c_i^\dagger\ket{0}\bra{0}\left(\prod_{j=1}^{n_0}c_j^\dagger\right)^\dagger$, which yields $c_{0,k}=\delta_{k,n_0}$ for different $n_0$. The results show that $\exp(n\gamma t)\braket{N}_t$ exhibits different polynomial scalings, depending on the initial state, whereas $1/\braket{\tilde{N}}_t$ for all initial states displays a linear scaling with respect to the evolution time $t$.


\begin{figure}
    \centering
    \includegraphics[width=1\linewidth]{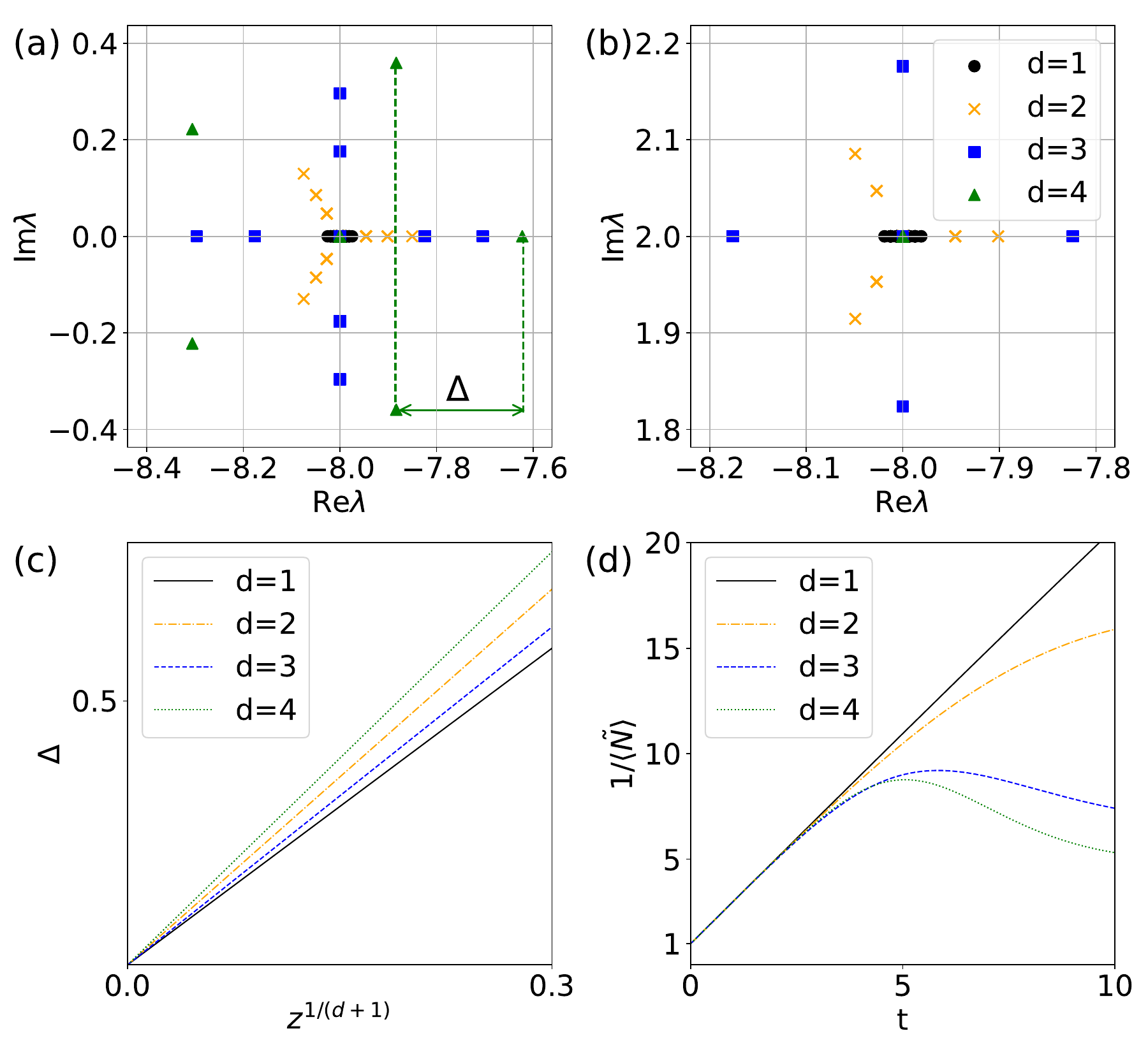}
    \caption{Liouvillian spectrum and dynamics subject to perturbations of the form $z\mathcal{L}_d$. We set the system parameters as $t=1,\gamma=2,n=4$. (a)(b) The perturbed spectra of the hybrid-Liouvillian,  $\mathcal{L}_H'=\mathcal{L}_H+z\mathcal{L}_d$ with $z=10^{-5}$, in the vicinity of (a) $\lambda=-8$ and (b) $-8+2i$, respectively, under different perturbations characterized by $d$. As an illustration, in panel (a), we mark the Liouvillian gap $\Delta$ for $d=4$.
    (c) The Liouvillian spectral gap $\Delta$ of the perturbed hybrid-Liouvillian spectra as functions of $z^{1/(1+d)}$.
    (d) Dynamics of  $1/\braket{\tilde N}_t$ under perturbations with different $d$ but a fixed $z=10^{-5}$.
   The system is initialized in $\rho(0,n_0=1)=c_1^\dagger\ket{0}\bra{0}c_1$.}
    \label{Fig2}
\end{figure}

\subsection{Perturbation near the higher-order EP}

A distinctive feature of EPs is the fractional power-law scaling of eigenvalues under a small perturbation in their vicinity.
With the presence of the higher-order Liouvillian EPs in our model, we show that different forms of perturbation can lead to distinct power-law scalings in the Liouvillian eigenvalues and the spectral gap, a phenomenon consistent with the behavior of higher-order EPs in non-Hermitian systems.

For this purpose, we vectorize the density matrix, so that the Liouvillian in Eq.~(\ref{eq:Lex}) becomes a non-Hermitian operator, given by
\begin{align}
    \mathcal{L}_H=-i(H\otimes I-I\otimes H^T)+\gamma\sum_j c_j\otimes c_j^*-n\gamma.
\end{align}
Let $J_0$ denote the largest Jordan block of $\mathcal{L}_H$, and $J_\perp$ the complementary block. The { Jordan form} of $\mathcal{L}_H$ in the fermion representation can then be expressed as
\begin{equation}
    J=P^{-1}\mathcal{L}_HP=\begin{pmatrix}
        J_0 & \\
        & J_\perp
    \end{pmatrix},\label{eq:J}
\end{equation}
where $P=(R_0,R_\perp)$, $P^{-1}=(L_0,L_\perp)^T$, $R_0$ and $L_0$ are the right and left eigenmatrices of $J_0$, respectively, while $R_\perp$ and $L_\perp$ are the right and left eigenmatrices of $J_\perp$.

We denote $\bm\xi_d=(\xi_1,\xi_2,\cdots,\xi_d)$ with $\xi_1<\xi_2<\cdots< \xi_d$ (for $d=1\cdots,n$), and introduce quantum jump operators of the form $L_{\bm\xi_d}=\sqrt{\gamma}c_{\xi_1}^\dagger c_{\xi_2}^\dagger\cdots c_{\xi_d}^\dagger$. In the following, we first show that the perturbed spectrum of $J_0$ exhibits a fractional power-law scaling $z^{1/(d+1)}$, under the perturbation 
\begin{equation}
    z\mathcal{L}_d\rho=z\sum_{\bm\xi_d}L_{\bm\xi_d}\rho L_{\bm\xi_d}^\dagger.\label{eq:perturb}
\end{equation}
Note that, while the vectorization of the perturbation term is $z\mathcal{L}_d=z\sum_{\bm \xi_d} L_{\bm\xi_d}\otimes L_{\bm\xi_d}^*$, the perturbation applied to the Jordan block $J$ is 
\begin{align}
    P^{-1}z\mathcal{L}_d P=z\begin{pmatrix}
        L_0^T \mathcal{L}_d R_0  & L_0^T \mathcal{L}_d R_\perp \\
        L_\perp^T \mathcal{L}_d R_0 & L_\perp^T \mathcal{L}_d R_\perp
    \end{pmatrix}.
    \label{zL_d}
\end{align}

In order to establish the relationship between perturbation $z\mathcal{L}_d$ and the generalized eigenmatrices of $\mathcal{L}_H$, we introduce the basis vectors of the fermionic Fock space, and express $L_{\bm\xi_d}$ and $z\mathcal{L}_d$ in this basis.
Denoting $\bm{\zeta}_m=(\zeta_1,\zeta_2,\cdots,\zeta_m)$, we have the Fock-space basis $\ket{\bm\zeta_m}=c_{\zeta_1}^\dagger c_{\zeta_2}^\dagger\cdots c_{\zeta_m}^\dagger\ket{0}$. 
Considering the completeness condition $\sum_m\sum_{\bm \zeta_m}\ket{\bm \zeta_m}\bra{\bm \zeta_m}=I$, $L_{\bm\xi_d}$ can be expressed as
\begin{align}
    L_{\bm\xi_d} &= \sum_m \sum_{\bm\zeta_m} L_{\bm\xi_d}\ket{\bm\zeta_m}\bra{\bm\zeta_m} \notag \\
    &=\sqrt{\gamma}\sum_m\sum_{\bm\zeta_m\cap \bm\xi_d=0}(-1)^{\sigma(\bm\xi_d,\bm\zeta_m)}\ket{\bm\xi_d\cup \bm\zeta_m}\bra{\bm\zeta_m},
\end{align}
where $\sigma(\bm \xi_d,\bm \zeta_m)$ is the sign factor that arises from the anti-commutation relations.

To evaluate Eq.~(\ref{zL_d}), we note that 
\begin{align}
    \mathcal{L}_d r_{0,k} & =\gamma\sum_{\bm \xi_d}\sum_{m,m'}\sum_{\bm\zeta_ m\cap \bm\xi_d=0}\sum_{\bm\eta_ {m'}\cap\bm\xi_d=0}\notag \\
&\quad\quad\ket{\bm\zeta_m\cup\bm\xi_d}\ket{\bm\eta_{m'}\cup\bm\xi_d}\bra{\bm\zeta_m}\bra{\bm\eta_{m'}}r_{0,k}\notag \\
    &=\gamma \sum_{\bm\xi_d}\sum_{\bm\zeta_k\cap \bm\xi_d=0}\ket{\bm\zeta_k\cup\bm\xi_d}^{\otimes 2}\bra{\bm\zeta_k}^{\otimes 2}r_{0,k}\notag \\
    & \sim \sum_{\bm\zeta_{d+k}}\ket{\bm\zeta_{d+k}}^{\otimes 2} \sim r_{0,d+k}.
\end{align}
It follows that the diagonal block $z(  J_{0,d})_{j,k}:= z(L_0^T \mathcal{L}_d R_0)_{j,k} \sim l_{0,j}^T r_{0,d+k}=\delta_{j,d+k}$, whereas the off-diagonal block vanishes $ (L_\perp^T \mathcal{L}_d R_0)_{j,k} \sim l_{\perp,j}^T r_{0,d+k}=0$ for all $j,k$, because of the biorthogonal relations of the generalized eigenmatrices.
Similarly, we have $ L_0^T \mathcal{L}_d R_\perp=0$, which leads to a block-diagonal perturbation matrix
\begin{equation}
    P^{-1}z\mathcal{L}_d P = z\begin{pmatrix}
           J_{0,d} & \\
          &   J_{\perp,d}
    \end{pmatrix},
\end{equation}
where we denote $  J_{\perp,d}=L_\perp^T \mathcal{L}_d R_\perp$. Here the explicit form of $  J_{0,d}$ is
\begin{equation}
  J_{0,d}=\begin{pmatrix}
    & & & & &\\
    & & & & &\\
   \alpha_0 & & & & &\\
   &\alpha_1 & & & &\\
   & & \ddots & & &\\
   & & &\alpha_{n-d} & &
\end{pmatrix},
\end{equation}
where $\alpha_0,\alpha_1,\cdots,\alpha_{n-d}$ are constants. 
Denoting $\lambda$ and $\lambda_0$ as the eigenvalues of $J_0+z  J_{0,d}$ and $J_0$, respectively, the characteristic polynomial of the perturbed largest Jordan block $J_0+z  J_{0,d}$ is given by
\begin{align}
    &\det(J_0+z  J_{0,d}-\lambda I_{n+1})\notag \\
    \simeq&( -\lambda+\lambda_0)^{n-d}[(-\lambda+\lambda_0)^{d+1}+z\sum_i\alpha ].
    \label{determinant}
\end{align}
It follows that the eigenvalues of the largest Jordan block are modified by terms of the order $z^{1/(1+d)}$. Hence, depending on the form of the perturbation, the various fractional power-law scalings of our $(n+1)$th order Liouvillian EP can be recovered. 

More relevant to the system dynamics, we now examine the behavior of the perturbed spectral gap. In particular, we consider the contribution of $J_{\perp,d}$ and demonstrate that, for all $d$, the spectral gap of the perturbed Liouvillian scales as $z^{1/(1+d)}$, which is consistent with the scaling of $\lambda$ in Eq.~(\ref{determinant}).

First, in the case of $d=n$, we have $\mathcal{L}_{d=n}=\sum_{\bm \xi_n}L_{\bm \xi_n}\otimes L_{\bm \xi_n}^*=\gamma r_{0,n}l_{0,0}^T$, so that the perturbation applied to $J_{\perp}$ vanishes exactly, $  J_{\perp,n}:=L_\perp^T \mathcal{L}_{d=n} R_\perp=0$. The spectral gap thus arises from the contribution of $J_{0,n}$, which is $z^{1/(1+n)}$.
Then, consider the special case of $d=1$. The perturbation is quadratic $\mathcal{L}_{d=1}\rho=\gamma\sum_jc_j^\dagger\rho c_j$, which is analytically solvable and physically associated with imperfect post selection. As discussed in more detail in Appendix~\ref{appendixC}, all EPs disappear and the associated eigenvalues and hence the spectral gap change with a fractional scaling $z^{1/2}$.

For the intermediate regime $1<d<n$, an analytical solution of $J_{\perp,d}$ is unattainable.
We therefore numerically study the evolution of the Liouvillian spectrum and the scaling of the spectral gap.
In Fig.~\ref{Fig2}(a)(b), we show the perturbed Liouvillian spectra near $\lambda=-8$ and $-8+2i$, respectively, for $n=4$ and with different $d$. The splitting of the EPs under perturbation is clearly visible.
In Fig.~\ref{Fig2}(c), we demonstrate that the spectral gap indeed grows with $z^{1/(1+d)}$, as expected.

Physically, the highest-order EP of the system occurs at $\lambda_0=-n\gamma$, which corresponds to the vacuum state. 
Once the perturbation $z\mathcal{L}_d$ is applied, the degeneracy at the EP is lifted, and the vacuum state is no longer the quasi-steady state. The perturbation $z\sum_{\bm\xi_d}L_{\bm\xi_d}\rho L_{\bm\xi_d}^\dagger$ corresponds to quantum channels which create $d$ fermions simultaneously. Consequently, a larger $d$ results in a more pronounced shift of the eigenstates from the vacuum state. 
Furthermore, as discussed above, under $z\mathcal{L}_d$, the Liouvillian spectrum is no longer gapless, but acquires a gap $\Delta\sim z^{1/(d+1)}$. Consequently, the total particle number $\braket{\tilde{N}}_t$, evaluated with the normalized density matrix $\tilde\rho$, relaxes exponentially to the { perturbed quasi-steady state} with a time scale $\tau\sim 1/\Delta\sim z^{-1/(1+d)}$. In Fig.~\ref{Fig2}(d), we show the dynamics of $1/\braket{\tilde{N}}_t$ under perturbations with different $d$.
While the relaxation time is divergent in the absence of perturbation, it is no longer so under perturbations.
With a given small $z$, a larger $d$ leads to a larger Liouvillian gap, and hence more pronounced deviation from the linear scaling of $1/\braket{\tilde{N}}_t$ in time.

Finally, we note that, for more general forms of perturbative quantum jumps, one can express them in the basis of the fermionic Fock space. Similar to our discussions above, the order of the scaling can also be obtained.

\subsection{Experiment implementation}

{ Before discussing possible implementations, we first estimate the success probability of the post-selection procedure, which is the trace of the unnormalized density matrix. As demonstrated in Eq.~(\ref{eq:trace}), apart from the overall exponential decay, the trace of the unnormalized density matrix is further modified by a polynomial factor of the order $n$. The latter originates from the higher-order Liouvillian EP. Such a polynomial slow-down is most effective at high fillings. For instance, consider the fully-filled initial state $n_0=n$, we have
\begin{equation}
    P_{\text{post}}(t)=\mathrm{Tr}\rho(t)=e^{-n\gamma t}(1+\gamma t)^{n}.
\end{equation}

For experimentally relevant systems with only few lattice sites, this probability remains appreciable in the experimentally relevant window \(\gamma t\lesssim 1\). For instance, we have $P_{\text{post}}(\gamma t=0.5)\simeq 0.685$ and $P_{\rm post}(\gamma t=1)\simeq 0.293$ for \(n=4\); $P_{\text{post}}(\gamma t=0.5)\simeq 0.516$ and $P_{\rm post}(\gamma t=1)\simeq 0.117$ for \(n=7\); $P_{\text{post}}(\gamma t=0.5)\simeq 0.389$ and $P_{\rm post}(\gamma t=1)\simeq 0.046$ for $n=10$. In the same time window, 
the normalized observable already shows a clear signal reflecting features of the EP. Thus, although the post-selection probability becomes exponentially small in the thermodynamic or asymptotic long-time limit, a proof-of-principle observation in a finite-size system is still experimentally feasible.

We now briefly discuss how the corresponding post-selected dynamics may be implemented in such a few-site setting. As illustrated in Fig.~\ref{Fig3}, to realize the complementary jump process $L_\mu'=\sqrt{\gamma}\,c_\mu^\dagger$, one introduces an auxiliary fermionic mode $d_\mu$ for each $c_\mu$. The auxiliary modes are coherently coupled to states in a Markovian reservoir, and can locally decay to the corresponding fermion mode $c_\mu$. 
Adiabatically eliminating the auxiliary modes gives rise to jump operators of the desired form.


The post selection required in our hybrid Liouvillian amounts to retaining only those trajectories in which no quantum jumps of the kind \(L'_\mu=\sqrt{\gamma}c^\dagger_\mu\) occur. 
For a scalable many-body system, such a condition is demanding. 
However, for the few-site proof-of-principle demonstration discussed above, the post selection is also feasible in principle. 
Specifically, the Hamiltonian dynamics conserve the total particle number \(N\), while the jump operators \(L_\mu=\sqrt{\gamma}c_\mu\) and \(L'_\mu=\sqrt{\gamma}c^\dagger_\mu\) respectively decrease and increase the total particle number by one, through incoherent dissipative processes. 
Therefore, with the simultaneous detection of emitted excitations from the auxiliary modes and the particle number of the system, trajectories containing an upward particle-number jump can in principle be identified and discarded. 

We note that alternative implementation schemes also exist, depending on the choice of the complementary jump operators. One may also consider a quantum-circuit implementation, by rewriting the fermion model into a dissipative spin chain through the Jordan-Wigner transformation.}



\begin{figure}
    \centering
    \includegraphics[width=0.85\linewidth]{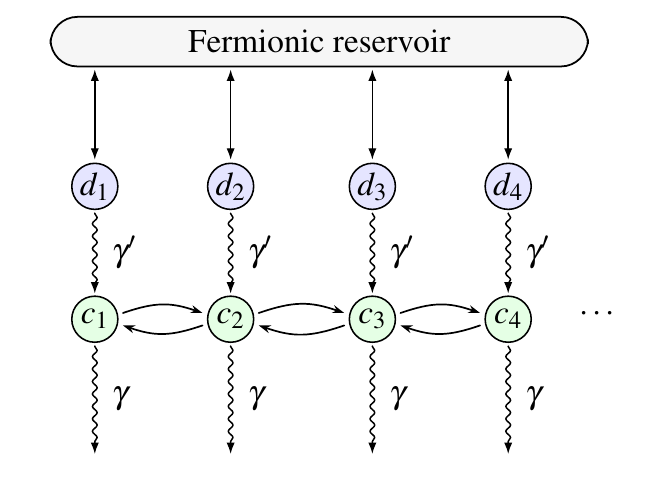}
    \caption{Schematic few-site implementation of the post-selected hybrid dynamics. The system fermions \(c_\mu\) form the system chain with loss channels $L_\mu=\sqrt{\gamma}c_\mu$. The complementary gain channel $L_\mu'=\sqrt{\gamma}c_\mu^\dagger$ can be engineered by coupling $c_\mu$ to an auxiliary fermionic mode $d_\mu$ through a dissipative jump. The auxiliary mode \(d_\mu\) is coupled to a fermionic reservoir that keeps it occupied.}
    \label{Fig3}
\end{figure}
 
\section{Conclusion}
We show that higher-order EPs can emerge in the Liouvillian eigenspectrum of dissipative fermions under partial post selection. Focusing on quadratic fermionic systems, we analytically solve the eigenspectrum, and show that the highest order of the Liouvillian EPs can approach the system size. We then demonstrate, both analytically and numerically, that the degeneracy at the EP can be lifted by introducing perturbations in the form of many-body quantum jump processes. 
Depending on the form of the perturbation, different fractional scalings of a higher-order EP can be observed, both in the Liouvillian spectrum and the system dynamics. 
We also emphasize that our main findings are based on hybrid-Liouvillians and cannot be reformulated directly for systems without post selection. Therein, steady states do not exhibit EPs, and the conditions for the shape matrix to be non-diagonalizable are more complicated.

Our work exemplifies the important role of EPs in many-body quantum open systems, and offers a physical context for the study of non-Hermitian physics in the many-body regime.


\begin{acknowledgments}
This work is supported by the National Natural Science Foundation of China (Grant No. 12374479), and by the Innovation Program for Quantum Science and Technology (Grant No. 2021ZD0301205).
\end{acknowledgments}

\appendix

\section{A proof that the matrices $T_\pm$ are nondiagonalizable.}
\label{appendixA}

In this appendix, we give a proof that the block matrices given in Eqs.~(\ref{T_p}) and (\ref{T_m}) are not diagonalizable.

\begin{lemma}
    Given a diagonalizable matrix $A$, the block matrix $M=\begin{pmatrix}
        A & B \\
        0 & A
    \end{pmatrix}$ is diagonalizable if there exists $X$, such that $B=[X,A]$.
\end{lemma}
\begin{proof}
    Consider a block similar transformation with $S=\begin{pmatrix}
        I & X\\
         & I
    \end{pmatrix}$ and $S^{-1}=\begin{pmatrix}
        I &-X\\
         & I
    \end{pmatrix}$. Suppose that $A$ is diagonalizable through an invertible matrix $P$, that is, $D=P^{-1}AP$ is a diagonal matrix. We then obtain
    \begin{equation}
        (SP)^{-1}MSP=\begin{pmatrix}
            D & P^{-1}(B+ AX-XA)P \\
             & D
        \end{pmatrix}.
    \end{equation}
    Thus, $M$ is diagonalizable if and only if $B=[X, A]$.
\end{proof}
\begin{theorem}
    $T_\pm$ are not diagonalizable.
\end{theorem}
\begin{proof}
    Here we focus on $T_-$, as the case for $T_+=-(T_-)^T$ is the same. According to Lemma 1, if $T_-$ can be diagonalized, there must exist a matrix $X$ such that $M=[H,X]$, which leads to $\mathrm{Tr}M=0$. However, $\mathrm{Tr} M=\sum_\mu\sum_i |l_{\mu,i}|^2$ is equal to zero only when all coefficients vanish $l_{\mu,i}=0$, which is a trivial case with no dissipation. Therefore, $T_-$ is not diagonalizable.
\end{proof}

\section{Properties of eigenmatrices}
\label{appendixB}

In this appendix, we show that for $\nu\neq0$, generalized eigenmatrices satisfy $\mathrm{Tr}r_{\nu,k}=\mathrm{Tr}Nr_{\nu,k}=0$, which enables us to evaluate the dynamics of the particle number Eq.~(\ref{eq:N}) in the main text. 

We first note that for a function $f$ of the total number of particles $f(N)$, if $f(N-1)$ is well-defined, we have
\begin{equation}
    \sum_j c_j^\dagger f(N)c_j= \sum_jc_j^\dagger c_jf(N-1)=Nf(N-1).
    \label{eq:b1}
\end{equation}
To derive the properties of eigenmatrices $r_{\nu,k}$, we multiply the generalized characteristic equations Eqs.~(\ref{Jordanchain1})-(\ref{Jordanchain3}) by $f(N)$ and take the trace. For Jordan blocks with eigenvalue $\lambda_\nu$ and size $m_\nu$, we obtain
\begin{align}
    &\mathrm{Tr}[Nf(N-1)r_{\nu,0}] - \frac{\lambda_\nu+n\gamma}{\gamma}\mathrm{Tr}f(N)r_{\nu,0}=0,\label{eq:b2}\\
     &\mathrm{Tr}[Nf(N-1)r_{\nu,k}] - \frac{\lambda_\nu+n\gamma}{\gamma}\mathrm{Tr}f(N)r_{\nu,k} \notag \\
     &\quad \propto \mathrm{Tr}f(N)r_{\nu,k-1},\quad k=1,2\cdots,m_\nu-1.\label{eq:b3}
\end{align}
In the above derivations, we have used $[H,f(N)]=0$ and Eq.~(\ref{eq:b1}). In the following, we consider various forms of the function $f(N)$ for the derivation.

We first consider eigenvalues $\lambda_\nu\neq -n\gamma$. Let $f(N)=\prod_{j=0}^k (N-j)$ ($k=0,1,\cdots,n$), for $r_{\nu,0}$, according to Eq.~(\ref{eq:b2}), we have the following equations
\begin{align}
\braket{N} & = \frac{\lambda_\nu+n\gamma}{\gamma}\mathrm{Tr}r_{\nu,0},  \\
\braket{N(N-1)} & = \frac{\lambda_\nu+n\gamma}{\gamma}\braket{N}, \\
&\cdots  \notag \\
\braket{\prod_{j=0}^{n} (N-j)} & = \frac{\lambda_\nu+n\gamma}{\gamma}\braket{\prod_{j=0}^{n-1} (N-j)}.
\end{align}
Note that $\prod_{j=0}^n(N-j)=0$, since $\lambda_\nu+n\gamma$ is nonzero, we have
\begin{align}
   & \braket{\prod_{j=0}^{n-1} (N-j)}=\braket{\prod_{j=0}^{n-2} (N-j)}\notag \\
   =&\cdots=\braket{N}=\mathrm{Tr}r_{\nu,0}=0.\label{eq:b7}
\end{align}
Similarly, by considering Eq.~(\ref{eq:b3}) for $f(N)$ above with different $k$, we have
\begin{equation}
    \mathrm{Tr} N r_{\nu,k} = \mathrm{Tr} r_{\nu,k}=0,
\end{equation}
where $k=0,1,\cdots,m_\nu-1$.

Next, we examine Jordan blocks associated with eigenvalues $\lambda_\nu=-n\gamma$ and $\nu\neq0$. For $r_{\nu,0}$, we set $f(N)=1$, and derive $\mathrm{Tr}Nr_{\nu,0}=0$. 
We then set $f(N)=(1+N)^{-1}$, in which case Eq.~(\ref{eq:b1}) cannot be directly applied, since $f(N-1)=N^{-1}$ is ill-defined. Note, however, that
\begin{equation}
    \sum_jc_j^\dagger (1+N)^{-1} c_j\ket{\bm\xi_k}=\begin{cases}
        \ket{\bm\xi_k}, \quad k\neq 0,\\
        0,\quad k=0.
    \end{cases}
\end{equation}
Hence,
\begin{equation}
    \sum_{j}c_j^\dagger (1+N)^{-1} c_j = I-\ket{\{0\}}\bra{\{0\}}=I-l_{0,0}.
\end{equation}
Consequently,
\begin{equation}
    \mathrm{Tr}[(1+N)^{-1}(\mathcal{L}_H+c)r_{\nu,0}]=\gamma\mathrm{Tr}(1-l_{0,0})r_{\nu,0}=0.
\end{equation}
Using the biorthogonality relation $\mathrm{Tr}(l_{0,0}r_{\nu,0})=0$, we obtain $\mathrm{Tr}r_{\nu,0}=0$. Proceeding similarly along the Jordan chain, one can show that, for each $k=0,1,\cdots,m_\nu-1$, we always have
\begin{equation}
    \mathrm{Tr} N r_{\nu,k} = \mathrm{Tr} r_{\nu,k}=0.
\end{equation}

\section{Exact result for $d=1$}\label{appendixC}

In the case of $d=1$, we have two sets of Lindblad operators $\{c_i\}$ and $\{c_i^\dagger\}$. Consider a more general perturbed hybrid-Liouvillian
\begin{equation}
    \mathcal{L}_H'\rho=-i[H,\rho]+\sum_{j,k}(M_{jk}c_j\rho c_k^\dagger+zM_{kj}^*c_k^\dagger\rho c_j)-\mathrm{Tr}M\rho.
\end{equation}
When $z=1$, the above equation reduces to the standard Lindblad master equation. The Hermitian matrix $M$ in the Majorana representation is given by
\begin{equation}
    M^M=\frac{1}{2}[M\otimes(I_2-\sigma_y)+zM^T\otimes(I_2+\sigma_y)].
\end{equation}
The shape matrix has a form similar to Eq.~(\ref{shapematrix})
\begin{equation}
    A'=T_1'\otimes I_2+T_2'\otimes i\sigma_y,
\end{equation}
where
\begin{align}
     T_1' &= \begin{pmatrix}
           -\frac{1}{2}h^{i} & \frac{1}{2}i(M+zM^T)\\
           - \frac{1}{2}i(M+zM^T)^T & -\frac{1}{2}h^{i}
       \end{pmatrix},\notag \\
       T_2'&= \begin{pmatrix}
           -\frac{1}{2}h^{r}  & -\frac{1}{2}(M-zM^T) \\
     -\frac{1}{2}(M-zM^T)^T& -\frac{1}{2}h^{r}
       \end{pmatrix}.
\end{align}
Following the derivation in the main text, we define the matrices $T_\pm'=T_1'\pm iT_2'$, with
\begin{align}
    T_+'&= \begin{pmatrix}
        -\frac{i}{2}h^T &ziM \\
        -iM ^T& -\frac{i}{2}h^T
    \end{pmatrix},\notag\\
    T_-'&=\begin{pmatrix}
\frac{i}{2}h&  iM\\
      -ziM^T  & \frac{i}{2}h
    \end{pmatrix}.
    \label{T_pm'}
\end{align}
The eigenvalues of $A'$ are thus the union of $T_\pm'$.

Moreover, we have $T_+'=-(T_-')^T$, which is consistent with the anti-symmetry of the shape matrix $A$. 
The matrices $T_\pm'$ are in general diagonalizable. More explicitly, in our example, $M=M^T=\gamma I_n$, thus the Fourier transform of $T_+'$ is
\begin{equation}
    T_+'(q)=-it\cos q+i\begin{pmatrix}
         &z\gamma\\
    -\gamma &
    \end{pmatrix},
\end{equation}
whose eigenspectrum is nondegenerate and given by $-it\cos q\pm\gamma z^{1/2}$. Thus, the hybrid-Liouvillian is diagonalizable, according to Eq.~(\ref{spectrum}). Given $2n$ binary numbers $\nu=(\nu_1,\nu_2,\cdots,\nu_{2n})$ with $\nu_i\in\{0,1\}$, we obtain the Liouvillian spectrum
\begin{align}
    \lambda_\nu=&2it\sum_{j=1}^n\cos\frac{2\pi j}{n}(\nu_{2j-1}+\nu_{2j})\notag \\
   & -z^{1/2}\gamma\sum_{j=1}^n(\nu_{2j-1}-\nu_{2j})-n\gamma.
\end{align}

We conclude that, under the perturbation with $d=1$, 
all EPs disappear, and the Liouvillian spectrum exhibits a gap of $\gamma z^{1/2}$. The maximum value of the real component of the eigenspectrum becomes $n\gamma(z^{1/2}-1)$, which reduces to $0$ for $z=1$ and recovers the case of the standard Lindblad master equation.


\end{document}